\documentclass[aps,prl,reprint,showpacs]{revtex4-2}
\usepackage{epsfig,amsmath}
\usepackage{verbatim}
\usepackage{amsthm}
\usepackage{amssymb}
\usepackage{bm}
\usepackage{latexsym}
\usepackage{color}

\usepackage{physics}
\usepackage{hyperref}

\newcommand{\beq}{\begin{equation}}
\newcommand{\eeq}{\end{equation}}
\newcommand{\bea}{\begin{eqnarray}}
\newcommand{\eea}{\end{eqnarray}}

\newtheorem{theorem}{Theorem}
\newtheorem{prop}{Proposition}

\begin{document}

\title{Quantum Resource Correction}

 \author{Mark S. Byrd$^1$}
% \author{Eric Chitambar}
 \author{Daniel Dilley$^2$}
 \author{Alvin Gonzales$^2$}
 \author{Zain Saleem$^2$}
 \author{Masaya Takahashi$^{1,}$$^{3,}$$^4$}
 \author{Lian-Ao Wu$^5$}

 \affiliation{$^1$School of Physics and Applied Physics, Southern Illinois University, Carbondale, Illinois 62901-4401, USA}
 
 \affiliation{$^2$Mathematics and Computer Science Division, Argonne National Laboratory}

 \affiliation{$^3$Center for Theoretical Physics, Polish Academy of Sciences,  Lotników 32/46, 02-668 Warszawa, Poland}
 \affiliation{$^4$Department of Physics, National Taiwan Normal University, Taipei 11677, Taiwan}
 \affiliation{$^5$Department of Physics, University of the Basque Country UPV/EHU and EHU Quantum Center, University of the Basque Country UPV/EH, 48080 Leioa, IKERBASQUE, Basque Foundation for Science,  48011, Bilbao Spain}
 
\begin{abstract}
Resource theories play a crucial role in characterizing states and properties essential for quantum information processing. A significant challenge is protecting resources from errors. We explore strategies for correcting quantum resources. We show that resource preserving operations in resource theory define a gauge freedom on code spaces, which allows for recovery strategies that can correct the resource while changing non-essential properties. This allows decoding to be simplified. The results are applicable to various resource theories and we provide an application to quantum sensing.
\end{abstract}

% \pacs{TBD}
\maketitle

%\tableofcontents

%---------------------------------------------------------------------------------------
%---------------------------------------------------------------------------------------

{\it Introduction}--A standard quantum error-correcting code (QECC) \cite{Knill_1997TheoryOfQECC,  Nielsen_1998InfoTheorApproachToQECAndReversMeas, gottesman_1997stabilizercodesquantumerror} enables the detection and correction of quantum errors by encoding logical information into a subspace of a larger physical Hilbert space. The primary goal is to restore the encoded quantum state to its original logical state. While this typically means restoring the state to an error-free version, exceptions exist. For instance, operator quantum error correction (OPQEC) permits variations in the state that do not affect the stored information \cite{Kribs_2005UnifiedAndGenApproachToQEC, Kribs_2006OperatorQEC, Bacon_2006OperatorQECSubsysForSelfCorrQuantMem, Poulin_2005StabFormForOperatorQEC, Nielsen_2007AlgebraicAndInfoCondForOperatorQEC}. OPQEC encompasses standard QECC and subsystem encoding methodologies \cite{Duan_1997_PreservingCohInQuantCompByPairQuantBits_DFS, Lidar_1998DFSForQuantComp, Zanardi_1997NoiselessQuantCodes_DFS, knill_1999theoryquantumerrorcorrectionForGenNoise_NS, Zanardi_2000StabQuantInfo_NS, Kempe_2001TheoryOfDecohFreeFaultTolUnivQuantComp_NS}.

Quantum coherence quantifies a state's superposition and along with entanglement enables much of the capabilities of quantum information processing such as quantum key distribution \cite{Bennett_2014BB84}, teleportation \cite{Bennett_1993Teleportation}, and prime factoring \cite{Shor_1997PrimeFactoring}. In this work, we adopt a broader perspective on entanglement \cite{Bennet_1996MixeStEntAndQEC, Vedral_1998EntanglementMeasuresAndPurificationProcedures} and coherence \cite{aberg2006quantifyingsuperposition, Baumgratz_2014QuantCoherence, Winter_2016OperationalResourceTheoryOfCoherence}, viewing them as resources within quantum resource theory \cite{Baumgratz_2014QuantCoherence, StreltsovRev:17, Chitambar_2019QuantumResourceTheor}. This framework analyzes information processing using a restricted set of permissible states and operations defined as free states and free operations, respectively \cite{Chitambar_2019QuantumResourceTheor}. Notably, incoherent states (those with zero coherence) 
are considered free states, and incoherent operations (which do not increase coherence) are categorized as free operations. 

In scenarios where errors occur, the objective shifts from state correction to resource correction, indicating that the original state may not be restored, while the resource is restored. The necessity of the resource alone for performing a specific task remains an interesting open question, dependent on the task itself. Still, we provide an application of quantum resource correction to quantum sensing. Intuitively, a smaller overhead is likely needed to restore a resource, rather than the state itself, during such a correction process. A previous investigation sought to preserve certain quantum properties, but their approach does not use encoding and instead relies on the use of a control Hamiltonian to preserve the property for as long as possible \cite{saurav_2025QuantumPropertyPreservation}. This strategy was utilized specifically for coherence \cite{Lidar_2005StabilizingQubitCoherenceViaTrackCntrl, hecht_2024beatingramseylimitsensingWithDetermQCtrl}. These approaches differ from the strategy in this paper, since we focus on encoding the states.

More specifically, the novel contributions of our work are as follows. We analyze and present necessary and sufficient conditions for logical resource correction in an information theory context for resources whose resource preserving operations are unitary. 
% While the resource correction conditions turn out to be similar to the QECC conditions, they allow for a simpler correction process. 
The logical resource preserving unitaries are not errors and therefore define a gauge freedom of the code. Thus, there are fewer possible logical errors. This provides an important simplification on the decoding problem, which is NP-complete or \#P-complete problem (depending on if degenerate decoding is used) \cite{Iyer_2015HardnessOfDecodingQuantStabCodes, deMarti_iOlius_2024DecodingAlgForSurfCodes}. Decoders determine the correction operations needed from the error syndrome and must balance decoding time and accuracy to minimize decoherence \cite{ Terhal_2015QECForQuantMems, Bravyi_2014EfficientAlgsFroMaxLikelihoodDecodingInTheSurfaceCode, Delfosse_2021almostlineartimeDecodingAlgForTopCodes, Fuentes_2021DegeneracyAndItsImpactOnTheDecodingOfSparseQC, Higgott_2025SparseBlossomCorrAMillErrsPerCoreSecondWithMWP}. The freedom in resource correction relaxes the constraints for the decoding problem and allows for great flexibility. Next, we apply these results to logical coherence correction. The $l_1$ norm of coherence is used, which is the sum of the absolute values of the off-diagonal elements of the density matrix for a fixed basis called the incoherent basis \cite{Baumgratz_2014QuantCoherence, Chitambar_2019QuantumResourceTheor}. Note that other measures exist such as the relative entropy of coherence \cite{Baumgratz_2014QuantCoherence,Yao_2016FrobeniusNormBasedMeasOfQuantcohAndAsym}.    

We also investigate leveraging the physical basis of the coherence measure and identify a class of channels that leave the coherence of any state from a subspace invariant. These resource-preserving channels also form a resource \cite{Hsieh_2020resourcePreservability}. For error correcting codes, the orthogonality of the error subspaces for a correctable error channel implies that the coherence does not change for any element of the code space when acted on by the discretized error channel. Thus, we can combine stabilizer measurements in many instances. This is possible while the number of physical errors is below the number of correctable errors. Next, we  investigate entanglement. The necessary and sufficient conditions also apply to entanglement. Lastly, we present explicit examples, including the Shor code, and a quantum sensing application of coherence correction, which highlight the distinctions and unique cases arising between quantum error correction and resource correction. While our examples focus on coherence and entanglement, similar methodologies can be extended to other quantum resources. 

{\it Resource Correction}--Throughout this paper, a free operation of a resource is defined as an operation which cannot create the resource. Let $Q$ be a proper resource quantifier which satisfies monotonicity under a free operation. We call a map resource preserving if it does not change the amount resource for all $\rho$ in the set of input states. Let $\mathcal{E}(\rho) = \sum_i E_i \rho E_i^\dagger$ be an error channel. When an operation $\mathcal{R}(\rho) = \sum_r R_r\rho R_r^\dagger$ satisfies the following condition, we say $R$ corrects the resource in the state $\rho $ measured by $Q$ against $\mathcal{E}$.  
\begin{align}\label{eq:equality}
    Q(\rho) -  Q\left(\sum_{r,\alpha} R_r E_\alpha \rho E_\alpha^\dagger R_r^\dagger\right) = 0.
\end{align}
After correction, the quantum state may differ by a resource preserving operation, but retains the same resource as the initial state before any errors occurred. 

% We investigate two scenarios: the resource correction and single-output resource correction. For the resource correction scenario we require that Eq.~\eqref{eq:equality} is satisfied $\forall \rho$ in the code space.
% In the single-output scenario, we require that the final state has the correct amount of resource in a single output, but the process of error and correction may yield a different output state for each run and thus the mixed outcome of states may have a different level of resource.

First, we present necessary and sufficient conditions for logical resource correction where the resource preserving operations   $\mathcal{F}$ are unitaries. Let a bar over a quantity denote the logical counterpart, e.g., $\ket{\bar{b}}$ is a logical state. Then $\bar{\mathcal{F}}$ is the set of logical unitaries that preserve the logical resource. The set of possible logical errors for resource correction is smaller than in QECC. Specifically, the set $\mathcal{O}/\bar{\mathcal{F}}$ forms a restricted subset of $\mathcal{O}$, where $\mathcal{ O}$ is the set of logical unitaries. 

In resource correction, the error is corrected iff
\begin{align}
    Q(\mathcal{R}\circ\mathcal{E}(\rho))=Q(\Phi(\rho)),\quad \forall \rho\in \mathcal{C},
\end{align}
where $\Phi$ is a resource preserving channel and $\mathcal{C}$ is the code space. This is less strict than QECC, where the error is corrected iff $\mathcal{R}\circ \mathcal{E}(\rho)=\rho, \forall \rho \in \mathcal{C}$. The goal in resource correction is to correct the output of the function $Q$.
\begin{theorem}\label{thm:logical_resource}
    Let $\mathcal{E}$ be a noise channel. The necessary and sufficient conditions for logical resource correction are
    \begin{align}\label{eq:thm_cond_qec}
        P E_i^\dagger E_k P=\alpha_{ik}\bar{U}_{\bar{F}}P\bar{U}_{\bar{F}}^\dagger, \forall \bar{U}_{\bar{F}}\in\bar{\mathcal{F}}
    \end{align}
    where $\alpha_{ik}$ is a Hermitian matrix and $P$ is the projector onto the code space. 
\end{theorem}

\begin{proof}
It is necessary that $\mathcal{R}\circ\mathcal{E}$ preserves the resource of any element of the code space. Thus, $\mathcal{R}\circ\mathcal{E}$ acts as a resource preserving unitary channel on the code space \cite{Peng_2016MaxCohStatesAndCoherencePreservOps}. Therefore, we have
\begin{align}
\mathcal{R}\circ\mathcal{E}(P\rho P)=\sum_{ij}R_jE_iP\rho P E_i^\dagger R_j^\dagger\\
\notag=\bar{U}_{\bar{F}}P\rho P\bar{U}_{\bar{F}}^\dagger, \; \bar{U}_{\bar{F}}\in \bar{\mathcal{F}}. 
\end{align}

Following standard quantum error correction arguments (Ref.~\cite{Nielsen_1998InfoTheorApproachToQECAndReversMeas} and Ref.~\cite{nielsen2011quantumCompAndQuantInfo} Chapter 10), we can use the unitary degree of freedom of CP maps and there exists complex numbers $c_{ij}$ such that 
\begin{align}
&R_{j}E_iP = c_{ji}\bar{U}_{\bar{F}}P \; \Rightarrow \;
PE_k^\dagger R_j^\dagger R_jE_iP =c_{jk}^*c_{ji}P
\end{align}
Summing over $j$ and using the fact that $\sum_jR_j^\dagger R_j=I$, we get
\begin{align}
    \label{eq:qecCond}P E_i^\dagger E_k P=\sum_jc_{ji}^*c_{jk}P=\alpha_{ik}P=\alpha_{ik}\bar{U}_{\bar{F}}P\bar{U}_{\bar{F}}^\dagger,
\end{align}
where $\alpha_{ik}$ is a Hermitian matrix. Since this is the QECC condition \cite{nielsen2011quantumCompAndQuantInfo} up to a resource preserving unitary, it is also sufficient.
\end{proof}

For resources, the correction is simplified. Consider the projector $P$. If we conjugate this with a $\bar{U}_{\bar{F}}\in \bar{\mathcal{F}}$, we have
\begin{align}
    \bar{U}_{\bar{F}}P\bar{U}_{\bar{F}}^\dagger=P, \forall \bar{U}_{\bar{F}}\in \bar{\mathcal{F}},
\end{align}
since it is a logical operation. In standard quantum error correction when $\bar{U}_{\bar{F}}$ is induced by the error channel, it results in a logical error. In contrast, in resource correction these do not result in a logical error, since elements of $\bar{\mathcal{F}}$ preserve the resource.

In this sense, $\bar{\mathcal{F}}$ defines a gauge freedom on the code space since these operations leave the resource invariant and form a group. This is important, since it allows for great flexibility in the decoding of the error syndrome, which is an NP-complete or \#P-complete problem \cite{Iyer_2015HardnessOfDecodingQuantStabCodes, deMarti_iOlius_2024DecodingAlgForSurfCodes}.

%  More concretely, consider the algebra of logical observables that are preserved by $\bar{\mathcal{F}}$
% \begin{align}
%     \mathcal{A}_R=\{\bar{O}|\bar{U}_{\bar{F}} \bar{O}\bar{U}_{\bar{F}}^\dagger=\bar{O}, \forall \bar{U}_{\bar{F}}\in \bar{\mathcal{F}}\}.
% \end{align}
% Resource correction only has to correct the algebra $\mathcal{A}_R$, whereas quantum error correction has to correct the full logical algebra of observables.

\textit{Logical Coherence Correction}--For the resource of coherence, we focus on the $l_1$ norm measure of coherence, $C(\rho)$ \cite{Baumgratz_2014QuantCoherence}, which is the sum of the absolute values of the off-diagonal terms of $\rho=\sum_{jk}\rho_{jk}\op{j}{k}$ in a fixed basis $S$, i.e.,
\begin{align}
    C(\rho)=\sum_{i\neq j} \abs{\rho_{ij}}.
\end{align}
The free unitary operations are 
\begin{equation}\label{eq:physical_free_u}
V_F =\sum_{j\in S}e^{-i\phi_j}\op{\pi(j)}{j}, \end{equation}
where $\pi(j)$ is a permutation and $\phi_j$ is a phase.  These permute the eigenvectors of the density operator and add an overall phase to each eigenvector. Free unitaries comprise the set of coherence preserving operations \cite{Peng_2016MaxCohStatesAndCoherencePreservOps}.  

Let the coherence be measured by the $l_1$ norm in a fixed logical basis given by $\bar{\mathcal{S}}$. Then the logical coherence is preserved by logical free unitaries \cite{Peng_2016MaxCohStatesAndCoherencePreservOps}.
% Let a bar over a quantity denote the logical counterpart, e.g., $\ket{\bar{b}}$ is a logical state. 
The set of logical free unitaries is defined as
\begin{equation}\label{eq:logical_free_u}
  \!\!\!  
  \bar{\mathcal{F}} = \{\bar{U}_{\bar{F},k}| \bar U_{\bar{F},k}\ket{\bar{j}} =e^{ia_{j,k}}\ket{\bar{b}_{j,k}},\forall(\ket{\bar j},\ket{\bar{b}_{j,k}})\in \bar{\mathcal{S}}\}, \!\!
\end{equation}
where $a_{j,k}$ is a real number. 
% In the context of physical operations applied to qubits, certain operators may exhibit coherence preservation only for states within the codespace. However, their effect on elements of the codespace can always be mathematically represented by a logical unitary transformation~\cite{Peng_2016MaxCohStatesAndCoherencePreservOps}. Consequently, such operations are formally equivalent to logical free unitary operations when restricted to the codespace.

For surface codes, correction is performed by connecting the measurement qubits that detect an error. In QECC, the length of these chains should be minimized to lower the probability of creating a logical error. In logical coherence correction, minimizing the length of these chains during decoding is not as important since a logical Pauli operator is a logical free unitary. More generally, in stabilizer codes, decoding simplifies as follows. Let $\mathcal{G}_n$ be the +1 elements of the $n$ qubit Pauli group, $s$ be the syndrome, the incoherent basis be the logical standard basis, and Pr$(.)$ the probability. The decoding problem becomes a problem of finding any element  
\begin{align}
    R\in\{E_i|\text{Pr}(E_i|s)>0, E_i\in \mathcal{G}_n\}
\end{align}
as the correction (see Ref.~\cite{deMarti_iOlius_2024DecodingAlgForSurfCodes} for comparison).

% When measuring the coherence of the encoded state, there are two types of coherence that can be measured: logical and physical coherence. Mathematically, the physical incoherent basis uses the physical structure of a logical basis, but extended to a basis over the full Hilbert space. Importantly, since the physical basis is just an extension of the logical basis, the two measures coincide on elements of the code space.
% The difference between these two measures is quantified by an example. 

% Let our code space be the span of the logical states
% $\ket{0_L}=\frac{1}{\sqrt{2}}(\ket{00}+\ket{11})$ and $\ket{1_L}=\frac{1}{\sqrt{2}}(\ket{00}-\ket{11})$. Suppose our state is
% \begin{align}
%     \rho_L=\op{0_L}=&\frac{1}{2}(\op{00}+\op{00}{11}\\
%     \notag&+\op{11}{00}+\op{11}).
% \end{align}
% To obtain the physical coherence, we calculate the coherence of the physical state $\frac{1}{2}(\op{00}+\op{00}{11}+\op{11}{00}+\op{11})$ in the standard basis, which results in a physical coherence of $1$. On the other hand, the logical coherence is calculated by taking the coherence of the logical state $\rho_L$ in terms of the logical basis, which results in a logical coherence of zero. Thus, the set of free unitaries for logical coherence and physical coherence are not necessarily equal.

\textit{Physical Coherence Correction}--When measuring the coherence of the encoded state, there are two types of coherence that can be measured: logical and physical coherence. Mathematically, the physical incoherent basis uses the physical structure of a logical basis, but extended to a basis over the full Hilbert space. Importantly, since the physical basis is just an extension of the logical basis, the two measures coincide on elements of the code space.

The case of physical coherence correction yields an interesting scenario. In QECC, information is typically only stored in the logical space. However, for the $l_1$-norm of coherence, the physical coherence measurement basis can be utilized.  The free unitaries for physical coherence for a fixed physical basis $S$ for measuring the $l_1$-norm of coherence is given by the set of free unitaries for the physical states, which are defined in Eq.~(\ref{eq:physical_free_u}).  
%\begin{align}
%    \mathcal{F}'\equiv\{U_{F,k}| U_{F,k}\ket{j} =e^{ia_{j,k}}\ket{b_{j,k}},\forall(\ket{j},\ket{b_{j,k}})\in \mathcal{S}'\}.
%\end{align}

Let the logical basis for the code space $\mathcal{C}$ be $\text{basis}(\mathcal{C})\subset \mathcal{S}$. Since $\mathcal{C}$ is a subspace,  the physical coherence preserving operations for $\mathcal{C}$ are not necessarily physical free unitaries. The channels corresponding to physical coherence preserving operations depend on the structure of the code space. For example, let the code space $\mathcal{C}$ be the 3 qubit repetition code and $\mathcal{S}$ be the standard basis. Then the bit-flip channel
\begin{align} \label{Eq:bf_Channel3}
    \mathcal{E}(\rho)=(1-p)I\rho I+ \frac{p}{3}\sum_{k = 1}^3 X_k \rho X_k    
\end{align}
preserves the coherence of any state in the code space due to orthogonality of the error subspaces. The coherence of this state is further discussed in the supplementary material \cite{SuppMat}. This example motivates the following result.
\begin{prop}\label{prop:physical_coherence}
    Let the noise channel be a Pauli channel $\mathcal{E}(\rho)=\sum_i\alpha_iE_i\rho E_i^\dagger$, the code space 
    $\mathcal{C}$ be a stabilizer code able to correct up to $t$ errors, and the Kraus operators $E_i$ act nontrivially on up to $t$ qubits. For the $l_1$ norm measure of coherence $C$, $C(\rho)=C(\mathcal{E}(\rho))$ $\forall \rho\in\mathcal{C}$, for the fixed physical incoherent basis
    \begin{align}
        \mathcal{S}=\{\ket{\psi} | \exists P_j, \ket{\psi} = P_j\ket{\bar{0}} \},
    \end{align}
    where $\ket{\bar{0}}$ is the logical 0 state of the stabilizer code and the $P_j$ are tensor products of the Pauli operators and identity.
\end{prop}
\begin{proof}
    Since the stabilizer code can correct up to $t$ errors and the noise is Pauli, each of the terms $E_i\rho E_i^\dagger$ are orthogonal $\forall \rho\in\mathcal{C}$. Since the map is also trace preserving, the sum of the absolute values of the off diagonal terms give the same coherence value for the initial and output states. 
\end{proof}

More generally, a class of channels that preserves physical coherence for states in a subspace $\mathcal{C}$ and basis $\mathcal{S}$ are channels $\sum_ip_iU_{F,i}^{\phantom{\dagger}}\rho U_{F,i}^\dagger$, where each of the terms are orthogonal, 
\begin{align}
&U_{F,i}^{\phantom{\dagger}}\rho U_{F,i}^\dagger U_{F,j}^{\phantom{\dagger}}\rho U_{F,j}^\dagger=0 \\
\text{such that } &U_{F,i}, U_{F,j}\in \mathcal{F}, \forall i\neq j, \forall\rho\in\mathcal{C}.
\end{align} 
The orthogonality of each of the terms imply that the off-diagonal terms do not overlap and the total coherence is preserved. QECCs by construction can achieve this channel by projecting onto the code space (Ref.~\cite{nielsen2011quantumCompAndQuantInfo} Chapter 10). For an $[[n, k, d]]$ code, the orthogonality of the error subspaces are preserved up to errors on $t\leq\frac{d-1}{2}$ qubits. This observation leads to a simplified coherence correction process for QECCs. 

% \textit{V. Physical Coherence Correction}.--
Consider an $[[n, k, d]]$ stabilizer code for multiple rounds of error and correction. We can combine the measurements of stabilizers for rounds up until $t$ errors accumulate. At this point, we either use the state or we measure all the stabilizers independently to discretize the errors. After discretizing the errors, we can use the rotated frame \cite{Knill2005_QuantCompWithRealisticNoisyDevice} and measure combined stabilizers again. We will clarify this with examples below. Note that in standard QECC we can also let errors accumulate by tracking the Pauli frame, but eventually we would have to correct the errors in order to use the state. In contrast, for physical coherence correction the state always has the right amount of coherence in the physical basis and can be used directly.

% %-----------------------------------------------------------------------------------------------------------------------------------------------------------------------------------------------
{\it Example 1: Physical Coherence, Repetition Code}--Let the code be the repetition code of distance 5, $\mathcal{S}$ be the standard basis, and the noise channel be an independent and identically distributed bit flip channel
\begin{align}
    \mathcal{E}(\rho)=(1-p)I\rho I+ \frac{p}{5}\sum_{i=1}^5 X_i\rho X_i.    
\end{align} Note that the single-qubit bit flip channel $\mathcal{E}(\rho)=(1-p)\rho + pX\rho X$ acting on a single qubit state can change the coherence of the input state. The 5-qubit repetition code can correct up to 2 errors. Thus, we can wait two cycles of errors before measuring the parities to discretize the errors. Similar to the discussion around Eq.~\eqref{Eq:bf_Channel3}, at each cycle, the state has the right amount of coherence in the standard basis and can be used directly. By direct calculation, we can see the sum of the absolute value of off-diagonal terms will not depend on the error rate, $p$, after the first and second cycle.

{\it Example 2: Physical Coherence, Shor Code}--Let the code space be the Shor code and coherence in the basis
\begin{align}
    \mathcal{S}=\{\ket{\psi} | \exists P_i, \ket{\psi} = P_i\ket{\bar{0}} \}    
\end{align}
where the $P_i$ are tensor products of the Pauli operators and identity and $\ket{\bar{0}}$ is the logical 0 state. Let the noise channel be an i.i.d. arbitrary single qubit error channel. The stabilizer generators \cite{nielsen2011quantumCompAndQuantInfo} are
\begin{align}
    \begin{bmatrix}
        Z & Z & I & I & I & I & I & I & I\\
        I & Z & Z & I & I & I & I & I & I\\
        I & I & I & Z & Z & I & I & I & I\\
        I & I & I & I & Z & Z & I & I & I\\
        I & I & I & I & I & I & Z & Z & I\\
        I & I & I & I & I & I & I & Z & Z\\
        X & X & X & X & X & X & I & I & I\\
         I & I & I & X & X & X & X & X & X  
    \end{bmatrix}.
\end{align}
If we are going to use the state after one round, we can combine the $Z$ stabilizer measurements and measure $ZZIZZIZZI$ and $IZZIZZIZZ$. This results in a mixed state with the correct amount of physical coherence. For multiple rounds, we must measure all the stabilizers independently until the last round. However, we never need to apply a recovering unitary; stabilizer measurements are sufficient.

{\it Entanglement Correction}--%Next, we discuss entanglement correction. 
Note that Theorem \ref{thm:logical_resource} also applies to entanglement. For entanglement, logical local unitary operations are entanglement preserving. Thus, as for logical coherence,  decoders need to solve a less stringent optimization problem. However, the previous discussion on physical coherence does not extend to entanglement, since the output state can be mixed, which would generally change the entanglement.

{\it Example 3: Entanglement}--To illustrate the formalism, consider the example of entanglement correction in a two-qubit logical state encoded in a surface code. Suppose a correctable error occurs. We can measure the stabilizers and use a fast decoder (e.g., Union-Find decoder \cite{Delfosse_2021almostlineartimeDecodingAlgForTopCodes, deMarti_iOlius_2024DecodingAlgForSurfCodes}). Importantly, in the context of entanglement correction, applying a logical Pauli operation during the recovery process is not considered an error, as it corresponds to an entanglement-preserving gauge transformation. Thus, the Union-Find decoder has the same accuracy as typically more accurate decoders, since not finding the lowest weight correction results in possibly constructing a logical Pauli operation.

{\it Application to Quantum Sensing (Coherence)}--Here we provide a physical use case of quantum resource correction in DC magnetometry. Consider a 5-qubit repetition code and the logical plus state
\begin{align}
    \ket{+}_L=\tfrac{1}{\sqrt{2}}(\ket{00000}+\ket{11111}).
\end{align}
In DC magnetometry with field $B_z$, the parameter $\alpha\equiv \omega t$ (with $\omega\propto B_z$) is imprinted by the Hamiltonian $H=\tfrac{1}{2}\sum_{k=1}^5 Z_k$, yielding the ideal state
\begin{align}
    \ket{+'}_L=\tfrac{1}{\sqrt{2}}(\ket{00000}+e^{i5\alpha}\ket{11111}).
\end{align}
Parameter measurement is accomplished by measuring $X^{\otimes 5}$.

Let the error channel act independently on each qubit $k$ as $x$-axis rotations,
\begin{align}
    \mathcal E(\rho)= p_0 I\rho I + \sum_{k=1}^5 p_k\,R_x^{(k)}(\theta_k)\rho R_x^{(k)}(\theta_k)^\dagger
\end{align}
where $\sum_{k = 0}^5 p_k = 1$ and the $\theta_k$ are rotation angles. Although the errors are continuous, measuring the code’s stabilizers discretizes them into single-qubit $X$ flips. Since the code corrects up to two $X$ errors, each of the error subspaces are orthogonal and the coherence relevant to $X^{\otimes 5}$ is preserved. Let $\mathcal{M}$ be the stabilizer-measurement channel. Then
\begin{align} \label{Eq:Exp_Val_Equality}
\tr\!\big(X^{\otimes 5}\,\mathcal{M}\!\circ\!\mathcal{E}(\op{+'}_L)\big)=\tr\!\big(X^{\otimes 5}\,\op{+'}_L\big),
\end{align}
so the metrological performance at the operating point is maintained without solving the decoding problem. Additionally, stabilizer measurements can be combined due to Prop.~\ref{prop:physical_coherence}.  Equality in Eq.~\eqref{Eq:Exp_Val_Equality} can be verified explicitly by considering the combined stabilizer measurements 
$S_1 = Z_1 Z_2 Z_3 Z_4$ and $S_2 = Z_2 Z_3 Z_4 Z_5$. Under these measurements, the map 
$\mathcal{M} \circ \mathcal{E}$ acting on the logical state $\op{+'}{+'}_L$ yields four possible output states:
\begin{align} \nonumber
    &\op{+'}{+'}_L, \quad 
    X_1 \op{+'}{+'}_L X_1, \quad
    X_5 \op{+'}{+'}_L X_5, \\
    &\text{and} \quad
    \frac{1}{p_2 + p_3 + p_4} \sum_{k=2}^4 p_k\, X_k \op{+'}{+'}_L X_k.
\end{align}
Using the linearity and cyclicity of the trace, one finds that each of these states results in the same expectation value, since
\begin{align}
    \tr\!\left[X^{\otimes 5} \, X_j \op{+'}{+'}_L X_j\right] = \tr\!\left[X^{\otimes 5} \op{+'}{+'}_L \right],
\end{align}
where the equality follows from $X_j X^{\otimes 5} X_j = X^{\otimes 5}$ for all $j \in \{1,2,3,4,5\}$. Hence, the expression in Eq.~\ref{Eq:Exp_Val_Equality} holds as claimed.

{\it Conclusion/Discussion}--In this work, we explored resource correction theory and presented applications in quantum information processing. Since resource correction focuses only on correcting the amount of resource of the state, it reduces the required overhead of standard QECC. For instance, stabilizer measurements can be combined. Additionally, the resource preserving operations result in great flexibility on the decoder and simplify the optimization problem. 

Our methods can minimize computational overhead in resource correction protocols by reducing the need for complex corrective operations. This is particularly advantageous for long-term quantum information storage, where preserving coherence and entanglement is critical, such as in memory qubits for quantum repeaters. Furthermore, our work offers a pathway for future research into coherence-preserving operations \cite{Hsieh_2020resourcePreservability}, which could be adapted as a broader resource-compatible transformation. 

In addition to the correction theory, we provided an application of coherence correction to quantum sensing.  We also note that, if the resource is a Bell state with known parity, but unknown phase, an unknown entangled resource state can be utilized for the Ekert91 quantum key distribution protocol \cite{Ekert_1991QuantCrypBasedOnBellsThrm}, since we can optimize over the labeling of the Bell test. 

Note that the necessary and sufficient conditions that were presented use information theoretic methods.  Thus, from that perspective, the conditions are necessary and sufficient.  However, information is not necessarily the objective.  In future work, we will explore more abstract theory and examples that do not rely on the information theoretic approach and provide a more general setting for resource correction.  

Finally, we note that there are many other applications to be explored since we do not know all quantum information processing protocols that utilize a resource rather than a definite state.  
Therefore, we believe that this work will lead to an interesting future line of investigation.  What are the measurements/gates required to preserve different resources?  This could serve as a cost for conversion among different resources (not just different amounts of the same resource.)

{\it Acknowledgments}--We thank Eric Chitambar and Emanuel Knill for helpful discussions. This research was supported in part by the InterQnet project at Argonne National Laboratory, administered by the United States Department of Energy. MSB acknowledges the support of the National Science Foundation's IR/D program.  The opinions and conclusions expressed herein are those of the authors, and do not represent any funding agencies. MT also acknowledges the financial support of the Polish National Science Foundation (NCN) through the grants Opus 2019/35/B/ST2/01896 and Quant-Era ``Quantum Coherence Activation By Open Systems and Environments'' QuCABOoSE 2023/05/Y/ST2/00139 and the support of the Taiwan National Science and Technology Council with the grant number NSTC 114-2811-M-003-025. LAW has also received support from the Spanish Ministry for Digital Transformation and of Civil Service of the Spanish Government through the QUANTUM ENIA project call - Quantum Spain, EU through the Recovery, Transformation and Resilience Plan-NextGenerationEU within the framework of the Digital Spain 2026. This material is based upon work supported by the U.S. Department of Energy, Office Science, Advanced Scientific Computing Research (ASCR) program under contract number DE-AC02-06CH11357 as part of the InterQnet quantum networking project.

\vfill

\small

\noindent\framebox{\parbox{0.97\linewidth}{
The submitted manuscript has been created by UChicago Argonne, LLC, Operator of 
Argonne National Laboratory (``Argonne''). Argonne, a U.S.\ Department of 
Energy Office of Science laboratory, is operated under Contract No.\ 
DE-AC02-06CH11357. 
The U.S.\ Government retains for itself, and others acting on its behalf, a 
paid-up nonexclusive, irrevocable worldwide license in said article to 
reproduce, prepare derivative works, distribute copies to the public, and 
perform publicly and display publicly, by or on behalf of the Government.  The 
Department of Energy will provide public access to these results of federally 
sponsored research in accordance with the DOE Public Access Plan. 
http://energy.gov/downloads/doe-public-access-plan.}}

\bibliography{coherencebib.bib}
\bibliographystyle{apsrev4-2}

\end{document}